\documentclass{llncs}
\usepackage[margin=1.5in]{geometry}

\usepackage{tikz}
\usepackage[colorlinks=true]{hyperref}

\usepackage[ n, operators, advantage, sets, adversary, landau,
probability, notions, logic, ff, mm, primitives, events, complexity,
asymptotics, keys]{cryptocode}
        
\usepackage{subcaption}
\usepackage[inline,shortlabels]{enumitem}

\usepackage[backend=biber,style=alphabetic,maxnames=10]{biblatex}
\let\ep=\varepsilon
\newcommand\X{\ensuremath{\mathcal X}} 
\newcommand\Y{\ensuremath{\mathcal Y}} 
\newcommand\Z{\ensuremath{\mathcal Z}} 
\newcommand\Ss{\ensuremath{\mathcal S}}
\renewcommand\S{\Ss}

\newcommand\D{{\mathcal D}} \newcommand\Oo{{\mathcal O}}

\newcommand\del{weakly deletion-compliant}
\newcommand\dels{weak deletion-compliance}

\newcommand\DELs{Weak Deletion-Compliance}

\renewcommand\Pr{{\rm Pr}}

\newcommand\N{\mathbb N}
\newcommand\naturals{\mathbb{N}}
\newcommand\bool{\{0,1\}}
\newcommand\joli[1]{{\cal #1}}
\newcommand\tounder[1]{\buildrel #1 \over \longrightarrow}

\def\PR#1{\PPR#1\fin}
\def\PPR#1,#2,#3\fin{\Big| \Pr\big[#1(#2)=1\big] - \Pr\big[#1(#3)=1\big]\Big|}

\title{Deletion-Compliance in the Absence of Privacy\thanks{This is the full
    version of \cite{9647774}.}}
\author{Jonathan Godin\inst{1}\thanks{This work was done while visiting the National Research Council Canada.} \and Philippe Lamontagne\inst{2}}
\institute{Université de Montréal, Montréal, Canada \and National
Research Council Canada, Ottawa, Canada}

\begin{document}

\maketitle

\begin{abstract}
Garg, Goldwasser and Vasudevan
(Eurocrypt 2020) invented the notion of deletion-compliance to formally model the ``right to be forgotten'', a concept that confers individuals  more control over their digital data. A requirement of deletion-compliance is strong privacy for the deletion requesters since no outside observer must be able to tell if deleted data was ever present in the first place.
Naturally, many real world systems where information can flow across users are automatically ruled out.\\
The main thesis of this paper is that deletion-compliance is a standalone notion, distinct from privacy. 
We present an alternative definition that meaningfully captures deletion-compliance without any privacy implications.
This allows broader class of data collectors to demonstrate compliance to deletion requests and to be paired with various notions of privacy. 
Our new definition has several appealing properties:
  \begin{itemize}
  \item It is implied by the stronger definition of \citeauthor{GGV} under
    natural conditions, and is equivalent when we add a strong privacy requirement. 
  \item It is naturally composable with minimal
    assumptions.
  \item Its requirements are met by data structure implementations that do not reveal the order of operations, a concept known as history-independence.
  \end{itemize}
  Along the way, we discuss the many challenges that remain in
  providing a universal definition of compliance to the ``right to be
  forgotten.''
\end{abstract}


\section{Introduction}
\label{sec:intro}

Our increasingly online lives generate troves of personal data from our digital interactions.
The widespread collection of this user-generated data for marketing and other purposes has lead in response to the concept of ``right to be forgotten''. This \emph{right} stipulates that service providers do not hold onto data beyond its useful lifetime, and that users should retain ownership and control over their data and may, for example, ask for all copies and derivatives of their data to be deleted. This right has been codified in a handful of legislatures: the European General data Protection Regulation~\cite{GDPR} enshrines into law the right to erasure, Argentina's courts have made decisions that support the right to be forgotten~\cite{carter_argentina}, and California enacted a law requiring businesses to comply with data deletion requests~\cite{CCPA}. However, laws are by their very nature vague and open to interpretation. 
 In response, researchers began to frame privacy laws in the formal language of cryptography~\cite{nissim_bridging_2017,cohen_towards_2020} to facilitate compliance. 

The present work continues this vein of research in the context of the ``right to be forgotten'', that was recently pioneered by \citeauthor{GGV} \cite{GGV}. They have defined a notion of \emph{deletion-compliance} that sets formal guidelines for service providers -- that we will hereby refer to as \emph{data-collectors} -- to satisfy in order to comply with data deletion requests in accordance with the law. Deletion-compliance is not in itself a formal representation of a law, but rather a sufficient criteria to comply with the spirit of the law. In the formalism of~\cite{GGV}, a data-collector $\X$ interacts with two other entities: the \emph{deletion-requester} $\Y$ that represents the set of users that request deletion of their data, and the \emph{environment} $\Z$ that represents other users and any other entity. Interactions with the data-collector are described via a \emph{protocol} $\pi$ to which is associated a corresponding deletion protocol $\pi_D$. \citeauthor{GGV} define deletion-compliance using the real vs ideal paradigm. In the \emph{real} execution, the parties $\X$, $\Y$ and $\Z$ interact as prescribed by $\pi$ and $\pi_D$ and by their own programs. The restrictions are that every protocol $\pi$ between $\Y$ and $\X$ is eventually followed  by the associated deletion protocol $\pi_D$, and that $\Y$ is not allowed to communicate to $\Z$. Compared to this real execution is an \emph{ideal} execution that differs by replacing $\Y$ with a \emph{silent} variant $\Y_0$ that does not communicate with $\X$ nor $\Z$. The triple $(\X,\pi,\pi_D)$ is said to be deletion-compliant if for any $\Y$ and $\Z$, the state of $\X$ and the view of $\Z$ resulting from a real execution cannot be distinguished from those resulting from an ideal execution.

One consequence of the above framework is that for a data-collector to be deletion-compliant, it must \emph{perfectly} protect the privacy of its users.
The view of the environment, an entity that captures all third
parties (other than the deletion-requester and the data-collector), is included in the
indistinguishability requirement. This forces strict
\emph{privacy} of the data entrusted to the collector. This property was endorsed by the authors of~\cite{GGV}: they argue that a data-collector that discloses information from its users loses control over that information and thus the ability to delete every copy.
However, one can see from the definition that leakage of a single bit of information to the environment -- whether $\Y$ is silent or not -- is sufficient to distinguish the real and ideal executions.
This criteria sets the bar very high in terms of privacy requirements.
For example, it automatically rules out deletion-compliance of data-collectors that reveal some user information \emph{as part of their core functionality}. This includes any service that involves the sharing of user-created content such as social networks, messaging apps, etc. Moreover, even well intended data-collectors can leak information in subtle ways. For example, machine learning models\footnote{The full version of~\cite{GGV} gives an example of deletion-compliant data-collector that can train and delete from a machine learning model, but it does not have a \emph{public} interface for making predictions with the machine learning model since it could leak data from $\Y$ to $\Z$.} reveal information at inference time about the training dataset~\cite{shokri_membership_2017,wang_beyond_2019}, unless protective measures are taken~\cite{dwork_calibrating_2006,al-rubaie_privacy-preserving_2019}.

Motivated by these observations, we propose an alternative \emph{weaker}, more realistic, definition that builds on top of their framework yet abstracts away any notion of privacy, focusing instead on capturing what it means to honestly comply with
deletion requests. It enables a broader class of data-collectors to demonstrate compliance to the right-to-be-forgotten and to be paired with any suitable notion of privacy. 
Our definition  is \emph{simulation-based} and inspired by the concept of \emph{zero-knowledge} for interactive proof systems \cite{goldwasser_knowledge_1989}. Intuitively, just as the existence of a simulator for the verifier in zero-knowledge proofs shows that the verifier \emph{gains} no knowledge from the proof, the existence of a simulator for the data-collector shows that the data-collector \emph{retains} no knowledge after deletion. A caveat to this analogy is that since we do not enforce privacy for the deletion-requester, the data-collector may retain some information on $\Y$ that it has learned through other means -- e.g. by interacting with the environment $\Z$. To model this possibility, we give the view of $\Z$ to the simulator that we task with producing the state of $\X$. We say that $(\X,\pi,\pi_D)$ is \emph{\del{}} if the simulator, given the view of $\Z$ as input, produces an output indistinguishable from the state of $\X$.

\paragraph{Limits of Cryptography for Modeling the ``Right to be Forgotten''.} There is a fundamental limit to the formal treatment of the concept of ``right to be forgotten''. Formal definitions can only capture the deletion of data for which the deletion requester can demonstrate \emph{ownership} over the data. One of the main motivators for such laws is to prevent forms of online abuse such as defamatory content or non-consensual sharing of adult content. The notion of \emph{dereferencing} -- the removal of data about an individual that was uploaded without consent -- falls outside the scope of this work. See Section~\ref{sec:related-work} below for references to work that propose solutions to this other problem.

Our work and that of \cite{GGV} assume honest behaviour of the data collector. Unless the data collector has a quantum memory~\cite{broadbent_quantum_2020}, this assumption is necessary when modelling compliance to data deletion request. We stress that our notion of deletion-compliance is independent of any notion of privacy, and in particular does not excuse the data collector from respecting the user's privacy.

\subsection{Contributions \& Comparison with \cite{GGV}}
\label{sec:comp-with}

Our work is in continuity with that of \citeauthor{GGV} and builds on the same framework. They argue that for a data collector to retain the ability to delete data, it must not share data with any other party that does not also provide guarantees for the deletion of this data.  Our definition demonstrates that this is more a design choice than a necessary condition, and that a meaningful notion of deletion-compliance exists in the absence of total privacy. Our weaker definition more closely matches the behavior that we expect from real-world data-collectors that comply with ``right to be forgotten'' laws.

Our approach has several advantages over the stronger definition of~\cite{GGV}, which from now on we refer to as \emph{strong deletion-compliance}.

\paragraph{Composability.}

A consequence of the strong privacy requirement is that delegating
data storage or computation becomes impossible. The solution
of~\cite{GGV} to this problem is to require that any third party $\X'$, which receives information on $\Y$ and is modeled as part
of the environment $\Z$, also respects some form of 
deletion-compliance.  The system as a whole then satisfies a notion called
\emph{conditional deletion-compliance}.  Intuitively, a data-collector
$\X$ that shares some data with $\X'$ is conditionally
deletion-compliant if $\X'$ is deletion-compliant for the
protocol describing the interaction of $\X$ and $\X'$.  This solution is
unsatisfactory for the main reason that it adds an assumption on the environment $\Z$ (containing $\X'$), which is modeled as an adversary whose goal it is to learn about $\Y$. In constast, our definition composes sequentially with no assumption on the environment or deletion-requester and only mild assumptions on the honest data-collector.

Parallel composition of strong deletion-compliance also suffers from the stark privacy requirement: it composes in parallel only for data-collectors and deletion-requesters \emph{that do not communicate with each other}. Our weaker definition is easily shown to compose in parallel (for either deletion-requester or data-collector) without additional assumptions.

\paragraph{Ease of Compliance.}

We argue that the definition of strong deletion-compliance is too restrictive for most real-world applications. We give examples of data-collectors, some natural and some designed to prove our point, that satisfy the \emph{spirit} of the ``right to be forgotten'' (and our definition), but that are not strongly deletion-compliant. 

\citeauthor{GGV} give as example a specific deletion-compliant data-collector built from a \emph{history independent} dictionary. An implementation of an abstract data structure (ADS) is history independent if any two sequences of operations that lead to the same state for the ADS also lead to the same memory representation of the implementation.
 In contrast, we show that this is a general property of our definition: if a data-collector is implemented using exclusively history independent data structures, then it is \del{}.  

Our simulation-based definition also paves the way for a \emph{constructive} way to prove compliance to data-deletion requests.  For example, the source code for a data-collector may be published together with the code for the corresponding simulator. An independent auditor could then run experiments to heuristically verify the claimed \dels{}.

\paragraph{Observations on Defining Deletion-Compliance.}

We elaborate and add to the discussions initiated by \cite{GGV} on issues that arise when attempting to define deletion-compliance. We highlight several subtleties that practitioners may face when trying to comply with (strong or weak) deletion-compliance. More precisely, we identify problems that occur when randomness is used in certain ways by the data-collector, and we point to subtle ways in which information on $\Y$ may remain in the data-collector post-deletion. We offer partial solutions to those issues and leave a definitive resolution to future work.

\subsection{Related Work}
\label{sec:related-work}

The formal treatment of data privacy laws using ideas inspired by cryptography has been investigated outside of this work and~\cite{GGV}. To the best of our knowledge, \citeauthor{nissim_bridging_2017}~\cite{nissim_bridging_2017} is the first work using the formal language of cryptography to model privacy laws written in an ambiguous legal language. They formalize the FERPA privacy legislation from the United States as a cryptographic game-based definition and prove that the notion of differential privacy~\cite{dwork_calibrating_2006} satisfies this definition. \citeauthor{cohen_towards_2020}~\cite{cohen_towards_2020} tackle this task of bridging legal and formal notions of privacy for the GDPR legislation~\cite{GDPR}. They provide a formal definition of the GDPR's concept  ``singling out'' and show that differential privacy implies this notion, but that another privacy notion, $k$--anonymity~\cite{Samarati98protectingprivacy}, does not.

The notion that programs and databases (unintentionally) retain information beyond its intended lifetime is not new. The following papers focus on the setting of users interacting with programs on their own machines and ensuring that data cannot be recovered if the machine is compromised. \citeauthor{stahlberg_threats_2007}~\cite{stahlberg_threats_2007} propose desired properties for forensically transparent systems where all data should be accessible through legitimate interfaces and not through forensic inspection of the machine state. They investigate common database implementations and find that they are vulnerable to forensic recoverability, then propose measures to reduce unintended data retention.
\citeauthor{kannan_making_2011}~\cite{kannan_making_2011} propose a technique to ensure that sensitive data does not linger in a program's memory state after its useful lifespan. They take snapshots of the application state and log events so that the program can be rewinded to a previous state and actions replayed to arrive at an ``equivalent'' state where the sensitive information is gone.
A survey paper by \citeauthor{reardon_sok_2013}~\cite{reardon_sok_2013} reviews methods for securely deleting data from physical mediums. As such they are less interested in removing from a system all traces of data from a particular user, but to ensure that data marked for deletion cannot be recovered with forensics.
\citeauthor{ArkemaSherr}~\cite{ArkemaSherr} introduce the concept of residue-free computing that provides any application with an \emph{incognito mode}  preventing any trace data from being recorded on disk. 
 
As previously mentioned,
our paper is only concerned with the case where the data being requested for deletion originated from the deletion-requester. 
\citeauthor{simeonovski_oblivion_2015}~\cite{simeonovski_oblivion_2015} propose a framework for deletion requests of an individual's data that has another source (e.g.\ news or social media). It automates the process of detecting personal information publicly available online, requesting deletion and proving eligibility of the deletion request. \citeauthor{derler_i_2019}~\cite{derler_i_2019} propose a cryptographic solution to data deletion in a distributed setting based on a new primitive called identity-based puncturable encryption.

\section{Preliminaries}
\label{sec:prelim}

Throughout the paper, we use $\lambda\in\naturals$ as a security parameter. We write $\ppt{}$ as shorthand for ``probabilistic polynomial time'' and we let $\negl[\lambda]$ denote an arbitrary negligible function, i.e. such that for every $k\in\naturals$, there exists $\Lambda>0$ such for any $\lambda\geq \Lambda$, $\negl[\lambda]< \frac 1{\lambda^{k}}$.

 For two families of random variables ${\cal A}=\{A_\lambda\}_{\lambda\in\naturals}$ and ${\cal B}=\{B_\lambda\}_{\lambda\in\naturals}$. We say that $\cal A$ and $\cal B$ are $\ep$--\emph{indistinguishable} (resp. computationally $\ep$--indistinguishable) and write $\mathcal{A} \approx^S_{\ep} \mathcal{B}$ (resp. $\mathcal{A} \approx^C_{\ep} \mathcal{B}$) to mean that for 
  every unbounded (resp. \ppt{}) distinguishers $\D$, for all $\lambda\in\naturals$
  \begin{equation}
  \def\PR#1{\PPR#1\fin}
  \def\PPR#1,#2,#3\fin{\Big| \Pr\big[#1(#2)=1\big] - \Pr\big[#1(#3)=1\big]\Big|}
    \PR{\D,A_\lambda,B_\lambda} \leq \ep(\lambda).\label{eq:3}
  \end{equation}
  We simply write $\mathcal{A}\approx^S\mathcal{B}$ (resp. $\mathcal{A}\approx^C \mathcal{B}$) and say $\cal A$ and $\cal B$ are (computationally) indistinguishable when $\ep$ is a negligible
  function of $\lambda$ and even omit the superscript when clear from the context.
We let $\Delta^S(\mathcal{A}, \mathcal{B})$  (resp.~$\Delta^C(\cdot,\cdot)$) denote the supremum of~(\ref{eq:3}) over the choice of unbounded (resp.~$\ppt{}$) distinguisher $\D$.
This notion satisfies the triangle inequality: $\Delta(\mathcal{A},\mathcal{B})\leq \Delta(\mathcal{A},\mathcal{C}) + \Delta(\mathcal{C} + \mathcal{B})$ for $\Delta\in\{\Delta^S,\Delta^C\}$. 

\subsection{The Execution Model}
\label{sec:model-execution}

We use essentially the same formalism as in~\citeauthor{GGV}. We review it here and modify it slightly to suit our needs.  Parties involved in an execution are modeled as \emph{Interactive Turing Machines} (ITM) with several tapes.
\begin{definition}[Copied from \cite{GGV}]
  An interactive Turing Machine (ITM) is a Turing Machine $M$ with the following tapes
  \begin{enumerate*}[(i)]
  \item a read-only \emph{identifier} tape;
  \item a read-only \emph{input} tape;
  \item a write-only \emph{output} tape;
  \item a read-write \emph{work} tape;
  \item a single-read-only \emph{incoming} tape;
  \item a single-write-only \emph{outgoing} tape;
  \item a read-only \emph{randomness} tape; and
  \item a read-only \emph{control} tape.
  \end{enumerate*}
  
  The \emph{state} of an ITM at any given point in its execution, denoted $state_M$, consists of the content of its work tape at that point. Its \emph{view}, denoted by $view_M$, consists of the contents of its input, output, incoming, outgoing, randomness and control tapes at that point. 
\end{definition}
ITMs are static objects that are instantiated by \emph{instances of ITMs} (ITIs). Throughout this paper, we may refer to ITIs simply as ``machines''. Each ITI is equipped with a unique identifier and may be under the control of another ITI, in which case the \emph{controlling} ITI can write on the \emph{controlled} ITI's control and input tape and read its output tape. ITIs can be engaged in a \emph{protocol} by writing on each other's incoming tapes in a way prescribed by that protocol. Any message an ITI writes on the incoming tape of another ITI is also copied on its own outgoing tape. A protocol describes the actions of each participating ITI and what they write on each other's incoming tapes. An instantiation of a protocol is called a \emph{session} uniquely identified by a unique session identifier $sID$. Each ITI engaged in a session of a protocol is attributed a unique party identifier $pID$ for that particular session. To each session of $\pi$ is associated a \emph{deletion token} that is a function of $sID$ and of the content of the incoming and outgoing tapes of the machines engaged in that session. The deletion token will be used as input to $\pi_D$ to ask for the deletion of the information stored during session $sID$ of $\pi$.

Note that the protocol $\pi$ can in reality represent many possible protocols $(\pi_1,\dots,\pi_n)$ by having the machine that initiates $\pi$ begin its first message with the index $i$ of which $\pi_i$ it wants to run. We can assume without loss of generality that $\pi$ is a two-round protocol, i.e. the machine that initiates $\pi$ first sends one message (writes to the other machine's incoming tape) and the other machine responds. To implement longer interactions, the machines can use the session identifier to resume the interaction later in an asynchronous manner.

\paragraph{Special ITIs.}

We consider three special ITIs: the data-collector $\X$, the deletion-requester $\Y$ and the environment $\Z$. $\Y$ represents the clients that will request deletion and $\Z$ represents the rest of the world and any other machines interacting with $\X$ that may or may not request deletion. Instead of directly engaging in protocol sessions with each other, these special ITIs create new ITIs under their control to interact in this session. This way, $\X$ has no way of knowing if it is interacting with $\Y$ or $\Z$. By abuse of terminology, we say that $\Y$ (or $\Z$) initiates a protocol with $\X$ if both machines instantiate ITIs that engage in a protocol session.  

The state or view of these special ITIs is the concatenated states or views of every ITI under its control. We assume w.l.o.g.\ that every ITI $\Z$ creates only engages in a single session of $\pi$ since $\Z$ can task these ITIs to write the content of all of their tapes on their output tape and $\Z$ can provide these tapes to newly instantiated ITIs.

\paragraph{Execution Phases.}

The execution happens in two phases. In the \emph{alive} phase, the special ITMs $\X$, $\Y$ and $\Z$ are instantiated with the security parameter $\lambda\in\naturals$ written on their input tape in unary and their randomness tapes initialized with a stream of random bits. The execution consists of a sequence of activations of $\X$, $\Y$ and $\Z$ and their controlled ITIs. An activated ITI runs until:
\begin{itemize}
\item It writes on the incoming tape of another machine; that machine then becomes activated.
\item It writes on its own output tape and halts; the machine that created that ITI becomes activated.
\item It creates a new ITI; that new ITI becomes activated.
\end{itemize}
Only one machine is activated at any single moment, the others are ``paused'', i.e. they do not read/write from/to any of their tapes. The order in which the ITIs are activated, with the exception of the above special cases, is under the control of the environment.

In the \emph{terminate} phase, each ITI created by $\Y$ that initiated a session of $\pi$ or $\pi_D$ with $\X$ is activated until it halts. For every session of $\pi$ for which the corresponding $\pi_D$ has not been initiated, an ITI is created by $\Y$ to initiate a session of $\pi_D$ with the corresponding deletion token and is executed until it halts.

\paragraph{Differences with \cite{GGV}.}

\citeauthor{GGV} consider two distinct executions: the real execution described above and an \emph{ideal} execution where $\Y$ is replaced with a \emph{silent} deletion-requester $\Y_0$ that does not initiate any protocol with $\X$. Since our definition is concerned with simulation and not with a real vs ideal scenario, we do not use this ideal execution, except when dealing directly with the definition of strong deletion-compliance.

Moreover, since the original definition of strong deletion-compliance requires the indistinguishability of the view of the environment in the real and ideal executions, $\Y$ is not allowed to write on the incoming tape of \Z's controlled ITIs. Our definition does not imply nor require privacy from the environment, therefore we lift that restriction.

We introduce the following new elements of notation. For two ITIs $M$ and $M'$, we denote by $(M,M')$ a new ITI that controls both $M$ and $M'$ (assuming they are not already under the control of another machine). Furthermore, we denote by $view_{M}^{M'}$ the subset of $view_M$ restricted to the interactions between $M$ and $M'$ (or of ITIs under their respective control).

\subsection{Strong Deletion-Compliance}
\label{sec:strong-delet-compl}

For completeness, we present the definition of strong deletion-compliance of~\cite{GGV}.
\begin{definition}[Strong Deletion-Compliance]\label{def:strong-delet-compl}
  Given a data-collector $(\X,\pi,\pi_D)$, an environment $\Z$ and a deletion-requester $\Y$, let $(state_\X^R,view_\Z^R)$ denote the state of $\X$ and view of $\Z$ in the real execution, and let $(state_\X^I,view_\Z^I)$ the corresponding variables in the ideal execution where $\Y$ is replaced with a special machine $\Y_0$ that does nothing -- simply halts when it is activated. We say that $(\X,\pi,\pi_D)$ is \emph{strongly statistically (resp. computationally) deletion-compliant} if, for all \ppt{} environment $\Z$, all \ppt{} deletion-requester $\Y$, and for all unbounded (resp. \ppt{}) distinguishers $\D$, there is a negligible function $\ep$ such that for all $\lambda\in\naturals$:
  \begin{equation}
    \label{eq:1}
    \left| \Pr[\D(state^R_\X,view^R_\Z)=1]- \Pr[\D(state^I_\X,view^I_\Z)=1] \right| \leq \ep(\lambda)\enspace.
  \end{equation}
\end{definition}

\section{A More Realistic Notion of Deletion-Compliance}
\label{sec:definition-dels}

We now present our proposed definition of deletion-compliance. The definition is \emph{simulation}-based, a data-collector $\X$ is weakly deletion-compliant if there is a simulator $\S$ that can produce the state of $\X$ from the view of the environment. 

\begin{definition}[\DELs]\label{def:weak del}
  Let $(\joli X,\pi,\pi_D)$ be a data-collector, $\Y$ be a deletion-requester and $\Z$ be an environment.  Let $state_{\joli X}$ and $view_{\Z}^\X$ respectively denote the state of $\joli X$ and the view of the interaction of $\Z$ and $\X$ after the terminate phase.
  We say that
  $(\joli X,\pi,\pi_D)$ is \emph{computationally (resp.~statistically)
    $\ep$--\del} if there exists a simulator $\Ss$ such that for all
  \ppt{} (resp.~unbounded) environment $\joli Z$, all \ppt{} deletion-requester $\joli Y$
  and all \ppt{} (resp.~unbounded) distinguishers $D$, there exists a
  negligible function $\ep$ such that for all $\lambda\in \N$
  \begin{equation}
    \Big|\Pr\big[ \D(state_{\X},view_{\Z}^{\X})=1\big]
      - \Pr\big[ \D\big(\S(view_{\joli Z}^\X),view_{\Z}^{\X}\big)=1\big]\Big| \leq \ep(\lambda)
    \end{equation}
    where the probability is over the random tapes of $\X,\Y,\Z$ and $\D$. We say that  $(\joli X,\pi,\pi_D)$  is computationally (resp.\ statistically) \del{} if it is computationally (resp.\ statistically) $\negl[\lambda]$--\del{}.

We define the \dels{} \emph{error} naturally as smallest $\ep\colon \N\to[0,1]$ such that $(\X,\pi,\pi_D)$ is $\ep$--\del{}.
\end{definition}

Intuitively, our definition asserts that the state of $\X$ after deletion can be completely described by its interaction with the environment. The key distinction with Definition~\ref{def:strong-delet-compl} is that we no longer require that the environment contains no information on $\Y$ by distinguishing between a real execution and an ideal execution. In essence, everything that $\X$ may retain about $\Y$ after deletion is something that is already known to the environment. By giving $view_\Z^\X$ as input to the distinguisher, we ensure that the simulator outputs a state for $\X$ that is consistent with the view of the environment.

For the purpose of transparency, we note that there are trivial ways to modify any data-collector $\X$ to satisfy our definition. For example, consider a data-collector that, in each session of protocol $\pi$, sends the content of its work tape to the other machine. Obviously this data collector is \del{} since a simulator has direct access to the state of $\X$ in the view of $\Z$.
However, since our definition is concerned with \emph{honest} data collectors and is intended to model real-world systems, we do not concern ourselves with this sort of worst case behavior of $\X$. Similarly, we do not worry about the case where $\Y$ sends all its data to $\Z$. Our definition quantifies over all $\Y$, so the simulator must work for any $\Y$ and in particular for honest $\Y$ (i.e. that does not conspire with the environment).

Observe that something similar can be said of the definition of strong deletion-compliance. A data collector that makes no effort to comply to deletion requests can be turned into one that is strongly deletion-compliant by keeping snapshots of its state after each protocol execution, rewinding to a previous state when a deletion request occurs and replaying the protocols that did not originate from $\Y$. This is similar to what~\cite{kannan_making_2011} propose for purging sensitive data from programs that were not designed with that goal in mind.

\paragraph{Deletion in Practice.}

Modern computer architectures are far from Turing machines. Deleting a file on modern operating systems simply removes any reference to that file's disk location from the file system. This is unsatisfactory from the point of view of deletion-compliance on many fronts. The file or fragments thereof can still be recovered until its location on disk is overwritten with new data. The common solution is to overwrite the disk location of the file with 0's or with random data upon deletion. However, there remains information in the disk on the \emph{size} of this file and on \emph{where} the data was stored, which can leak metadata on the nature and order of operations.  There are also technological challenges to applying such an approach to single files with the increasing popularity of SSD storage~\cite{wei_reliably_2011}. An alternative solution known as \emph{crypto-shredding} encrypts files with unique keys and overwrites the keys to delete the file. The resulting system would be computationally deletion-compliant, as opposed to statistically. However, we are left with a similar problem: the non-zero disk space of the data collector (i.e. the size of $state_\X$) scales with the number of protocols it ran. This leaks information on the deletion requester and thus is not (even computationally) deletion-compliant. It is an issue that affects both weak and strong deletion-compliance.

A potential remedy is to give the simulator a \emph{masked} version of the deletion-requester's view. The masked view would consist of the same elements as $view_\Y^\X$, but instead of having the contents of the incoming and outgoing tapes, it has only the length of this content in unary. A similar solution was already hinted at by~\citeauthor{GGV}, i.e. have a deletion-requester that sends the ``same kinds of messages to $\X$, but with different contents''. With such a masked view of $\Y$, the simulator could know the timing of operations done by $\X$ and recreate any \emph{gaps} in memory that resulted from deletion. We did not adopt this approach with our definition since doing so would make it incomparable to strong deletion-compliance.

\paragraph{Simulation and Randomness.}

How the data-collector uses randomness can have an influence on whether or not it satisfies Definition~\ref{def:weak del}. 
Consider for example a data collector that has the following protocols:
\begin{itemize}
\item $\mathtt{query}(1^\lambda)$: Read $\lambda$ bits of the randomness tape. Let $x$ be the value. Store $x$ in memory and return $y=f(x)$ where $f$ is a cryptographic one-way function.
\item $\mathtt{delete}(y)$: Search through memory to find $x$ such that $f(x)=y$. If found, delete $x$ from memory, otherwise do nothing. 
\end{itemize}
Obviously such a data collector meets the intuition of deletion-compliance, and indeed it does appear to satisfy the strong variant of deletion-compliance (see discussion below). However, in our simulation-based definition the simulator has no way, given the view of $\Z$ that consists of elements $y$ from the image of $f$, to efficiently recreate the state of $\X$ that consists of the preimages of the $y$'s without the random tape of $\X$ that was used to sample the $x$'s. Since the simulator must produce a state for $\X$ that is consistent with the view of $\Z$, the existence of an efficient simulator would contradict the one-wayness of $f$. 

There are two apparent solutions to the above problem. The first is to give $\S$ access to the random tape of $\X$ in addition to the view of $\Z$ to recreate the working tape of $\X$. The other is to \emph{not} give $\S$ access to the random tape of $\X$ and ask that the state of $\X$ can be efficiently computed only from the view of $\Z$. This limits the ways in which the data collector uses its ``private'' randomness beyond its initialization phase. We chose the second approach for the following reasons. First, it makes our definition closer, and thus more easily comparable, to strong deletion-compliance. Second, going with the first approach would give more power to the simulator, weakening further the definition. Third, most examples where this poses a problem like for the above data-collector can be remedied by employing \emph{public randomness}, e.g.\ the above data-collector would receive the string $x$ from the other machine in protocol $\mathtt{query}$.

We remark that a similar issue appears in~\cite{GGV}: in the real and ideal executions, different parts of the random tape of $\X$ may be used to run the corresponding protocols with $\Z$. For example, suppose that $\X$ is initialized with the same randomness tape in the real and ideal worlds. If in each execution of $\pi$, $\X$ reads $n$ bits of its randomness tape (i.e. moves the head $n$ positions to the right), then the position of the randomness tape head will differ for protocols $\pi$ initiated by $\Z$ in the real and ideal worlds. In the real world, $\Y$ initiated $\pi$ with $\X$ thus moving the head on the randomness tape and in the ideal world, $\Y_0$ does not initiate $\pi$ thus not moving the randomness tape head. This will result in different views for the environment and, in the above data collector, a different state for $\X$. Their solution\footnote{See Appendix A of their full version (\url{https://arxiv.org/abs/2002.10635}).} to this issue is to condition on the views being the same in the real and ideal case. If the values for the view of $\Z$ have the same domain in the real and ideal worlds, it is then sufficient to bound the distance between the states of $\X$ in the real and ideal cases that led to the same view for $\Z$. 

In general, this approach only works if $\X$ employs ``public'' randomness, i.e. if the views of $\Y$ and $\Z$ uniquely determine the random tape of $\X$ (even inefficiently so, as in the above data-collector). Otherwise, the real and ideal data-collector may use different private randomness leading to different states for $\X$, even if the view of $\Z$ is the same in both executions. For  example, considering the above data collector, if the function $f$ is not one-to-one (i.e.\ not a bijection), then conditioning on the real and ideal views of $\Z$ being the same no longer works since the real and ideal states for $\X$ may contain different $x,x'$ such that $f(x)=f(x')=y$ is in the view of $\Z$.
Along with this public randomness, the data-collector can use ``private'' randomness which is roughly defined as randomness that only affects the internal state of $\X$ and not the views of the machine it interacts with.

\subsection{The Relationship Between Strong and \DELs}
\label{sec:relat-betw-strong}

To substantiate our claim that the restrictions imposed by strong deletion-compliance (Definition~\ref{def:strong-delet-compl}) are too severe, we present a simple data-collector that demonstrates a separation between strong and \dels{}. This data-collector, albeit contrived, showcases the fact that leakage of a single bit of information on $\Y$ rules out the strong flavor of deletion-compliance, but not our weaker variant. The data-collector is described in Fig.~\ref{fig:leak-1bit}.

\begin{figure}[h]
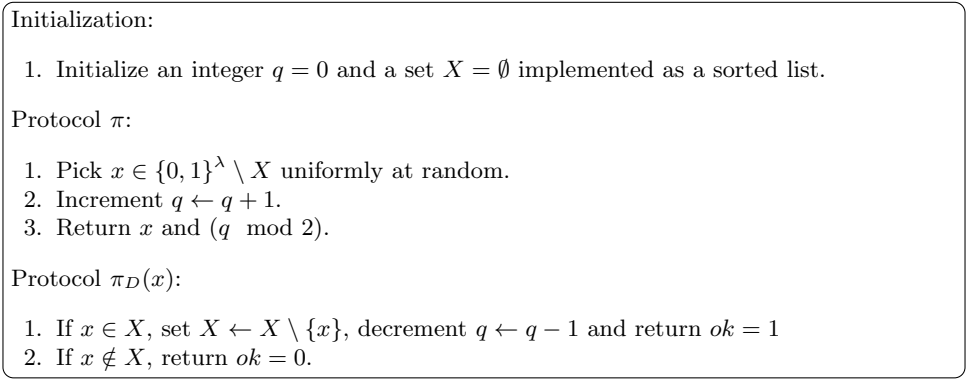
\centering

  \tikz{\node [draw,rounded corners] {
      \begin{minipage}{0.9\linewidth}
        Initialization:
        \begin{enumerate}
        \item Initialize an integer $q=0$ and a set $X=\emptyset$ implemented as a sorted list.
        \end{enumerate}

        Protocol $\pi$:
        \begin{enumerate}
        \item Pick $x\in\bool^\lambda\setminus X$ uniformly at random.
        \item Increment $q\leftarrow q+1$.
        \item Return $x$ and $(q\mod 2)$.
        \end{enumerate}

        Protocol $\pi_D(x)$:
        \begin{enumerate}
        \item If $x\in X$, set $X\leftarrow X\setminus \{x\}$,
          decrement $q\leftarrow q-1$ and return $ok= 1$
        \item If $x\notin X$, return $ok= 0$.
        \end{enumerate}

      \end{minipage}
    };}
\caption{The data-collector $\X$ of Theorem~\ref{thm:separation}.}
\label{fig:leak-1bit}
\end{figure}

\begin{theorem}\label{thm:separation}
  The data-collector $(\X,\pi,\pi_D)$ described in Fig.~\ref{fig:leak-1bit} is not strongly deletion-compliant, but it is \del{}.
\end{theorem}
\begin{proof}
  We first show that $\X$  is not strongly deletion-compliant. To do so, we construct an environment $\Z$ and a deletion-requester $\Y$ such that the view of $\Z$ in the real and ideal executions are perfectly distinguishable.

  The behavior of $\Y$ and $\Z$ in the alive phase is as follows:
  \begin{itemize}
  \item $\Y$: on first activation, initiate protocol $\pi$ with
    $\X$. On subsequent activations, do nothing and halt.
  \item $\Z$: activate $\Y$ once. After $\Y$ halts, initiate protocol $\pi$ with $\X$. When the instance of $\pi$ concludes (i.e. the ad-hoc client-side ITI halts), declare the end of the alive phase.
  \end{itemize}
  In the terminate phase, $\Y$ invokes $\pi_D$ with the deletion token it received from $\X$ during $\pi$ as is required from the execution model.

  Let  $state^R_\X$ and $view^R_\Z$ be the random variables describing the state of $\X$ and view of $\Z$ at the end of the above real execution (with $\Y$ acting as described). Let $(x, p)\in \bool^\lambda\times \bool$ be the message sent from $\X$ to $\Z$ in the instance of $\pi$. Then $state_\X^R= (X,q)=(\{x\}, 1)$ and $view_\Z^R= (x,p)= (x,0)$ since when $\Z$ initiates $\pi$, $\Y$ has already done its session of $\pi$ and has not yet done $\pi_D$, so $q=2$ during $\Z$'s instance of $\pi$, and $q=1$ after $\Y$ initiates $\pi_D$.
  Let $state^I_\X$ and $view^I_\Z$ be the same random variables in the ideal world where $\Y$ is replaced with $\Y_0$ that on every activation, does nothing and halts. Then $state^I_\X= (\{x\},1)$ is the same as in the real world, but $view^I_\Z=(x, 1)$ since $\Y_0$ does not initiate $\pi$ and $q=1$ in $\Z$'s session of $\pi$. The views of $\Z$ in the real and ideal executions can thus perfectly be distinguished.

  To show that $\X$ is \del{}, we must give a simulator that recreates the state of $\X$ from the view of $\Z$ (in the real execution). Let $\Y$ and $\Z$ be an arbitrary deletion-requester and environment where $\Z$ initiates $t=\poly[\lambda]$ sessions of $\pi$ or $\pi_D$ with $\X$. The view of $\Z$ contains of the list of protocols initiated with $\X$ and the contents of the incoming and outgoing tapes of the client-side ITIs for those protocols. Then we can parse the view of $\Z$ as $view_\Z= \{(\tau_i, x_i, p_i, ok_i)\}_{i\in [t]}$ where we set $p_i=\bot$ when $\tau_i=\pi_D$ and $ok_i=\bot$ when $\tau_i=\pi$. 

  We define the simulator $\S$ as follows. Initialize a set $S=\emptyset$ and integer $q=0$. For every  $(\tau_i, x_i, p_i, ok_i)\in view_\Z$, if $\tau_i=\pi$, insert $x_i$ in $S$ and increment $q$, and if $\tau_i=\pi_D$, if $x_i\in S$, remove $x_i$ from $S$ and decrement $q$, else do nothing. The output of $\S$ is $({\tt sort}(S),q)$. To prove indistinguishability, consider a deletion-requester $\Y$ that initiates $r=\poly[\lambda]$ sessions of $\pi$ (and an equal amount of $\pi_D$) with $\X$. 
  Let $\hat X$ be the $r$ deletion tokens that $\Y$ received from $\X$ in its session of $\pi$. We assume that $ x_i\notin \hat X$ for $i\in [t]$ and will add the probability of this event to the distinguishing advantage. Under this assumption, $state_\X$ is equal to $\S(view_\Z)$ since the set $X$ in the state of $\X$ contains exactly $S=\{x_i\}_{i\in[t]}$ (in sorted order) and $q=|S|$. Then the statistical distance between $state_\X$ and $\S(view_\Z)$ is the probability that $\hat X\cap\{x_i\}_{i\in[t]}\neq \emptyset$ which is at most $\frac{\poly[\lambda]}{2^\lambda}$ since $\Z$ is \ppt{} and can make at most $\poly[\lambda]$ queries to $\pi$ and $\pi_D$.
\qed  
\end{proof}

A desirable property of our definition of \dels{} is that it is implied by the stronger definition. There are some technical hurdles in proving that fact directly, however, mainly caused by the issue of randomness discussed earlier and the fact that strong deletion-compliance deals with two distinct real and ideal executions. The data collector does not start off with the same random tape in each executions, and even if it did, then it could result in observable differences in the views of the environment since randomness is consumed for the protocols initiated by $\Y$. We introduce an assumption that is used in this section only to prove that strong deletion-compliance implies \dels{} for the data collectors that satisfy this assumption.

The assumption is as follows. We distinguish between two forms of randomness used by $\X$. First, $\X$ is allowed to use \emph{private} randomness at initialization phase and during any protocol. Private randomness is defined such that changing the contents of the private random tape of $\X$ does not produce an observable change in the view of $\Z$ or $\Y$. The data collector $\X$ is also allowed a \emph{public} random tape with the restriction that the content of that tape can be efficiently reconstructed using $view_\Z$ and $view_\Y$.

\begin{definition}\label{def:public-private-tapes}
  We say that the machine $\X$ has \emph{private/public--separable} random tape if it can be divided into two random tapes $priv_\X$ and $pub_\X$ that satisfy the following.
  \begin{itemize}
  \item \emph{Private}: Let $M$ be an arbitrary ITI interacting with $\X$ and let $view_M(R)$ be the view of $M$ when interacting with $\X$ conditioned on $priv_\X=R$. Then for any values $R',R\in_R\bool^{\poly[\lambda]}$ for the private random tape $priv_\X$, $view_M(R)= view_M(R')$ with probability at least $1-\negl[\lambda]$.
  \item \emph{Public}: There exists an efficient function $f_\X(\cdot)$ such that $f_\X(view_\Z,view_\Y)$ returns the content of public random tape $pub_\X$.
  \end{itemize}
\end{definition}

Note that for every data collector given as example in (the full version of) \cite{GGV}, its random tape can be divided in such a way. Intuitively, we ask that the data collector $\X$ be deterministic up to some random initialization that only affects the internal state of $\X$, and that all further randomness can be efficiently computed from the views of machines interacting with $\X$. We show that under the above assumption, if $\X$ is strongly deletion-compliant, then it is \del{}.

\begin{theorem}\label{thm:ggv-implies-us}
  Let $(\X,\pi,\pi_D)$ be a strongly deletion-compliant data collector with public/private--separable random tape. Then $(\X,\pi,\pi_D)$ is \del{}. 
\end{theorem}
\begin{proof}
 To compare the definition of \dels{} with that of strong deletion-compliance, we need to clarify the execution context. In~\cite{GGV}, there are two parallel executions: the real execution where $\X$ interacts with $\Y$ and $\Z$ through protocols $\pi$ and $\pi_D$, and the ideal execution where $\Y$ is silent (i.e. it does not initiate any protocol). Moreover, in the formalism of~\cite{GGV}, $\Y$ is not allowed to send messages to $\Z$. For our definition, we are only interested in the real execution, where we introduce a simulator that reconstructs the state of $\X$ from only the view of $\Z$. We  will use the real/ideal indistinguishability to show that this simulator works. 
  
  We consider the following simulator $\joli S$. On input $view_\Z$,
  it runs a new instance of $\X$ with a freshly initialized private
  random tape $priv_\X$ and, given the view of the environment
  $view_{\joli Z}$ it computes a public random tape $pub_\X$ of $\X$
  consistent with that view (where the view of $\Y$ is assumed to be
  empty).  It then simulates the execution of every protocol $\pi$ and
  $\pi_D$ in the same order they appear in $view_\Z$ by using the
  values of the outgoing messages of $\Z$ as incoming messages for $\X$ for every such
  protocol and using the random tapes $priv_\X$ and $pub_\X$ as the
  sources of randomness. After simulating every protocol in $view_\Z$, it outputs the state of its instance of $\X$.

 Intuitively, the simulator can always simulate $\X$ in this way because the state of $\X$ is a deterministic function of its random tapes and of the views of $\Y$ and $\Z$. Since we may assume that $\Y$ is silent, by virtue of the indistinguishability of the real and ideal executions, the state of $\X$ is function only of $priv_\X$, $pub_\X$ and $view_\Z^I$ to which $\S$ has access to. By Definition~\ref{def:public-private-tapes}, since $priv_\X$ has no impact on the view of $\Z$ and since $pub_\X$ is computed from this view, the above simulation will also produce responses from $\X$ that are consistent with the view of $\Z$.
  By the above observations, 
  this simulation exactly replicates the state of $\X$ in the ideal
  execution where $\Y$ really is silent: $\joli S(view_{\joli Z}^I)= state_{\joli X}^I$.

  The rest of the proof is the same whether we consider statistical or
  computational indistinguishability, so we use the notation
  ``$\approx$'' to represent either.  We now compare
  $\S(view_{\joli Z}^I)$ and $\S(view_{\joli Z}^R)$. By assumption, we
  have that
    \begin{equation*}
      (state_\X^R,view_\Z^R)\approx (state^I_\X, view^I_\Z)\enspace.
    \end{equation*}
  Seeking a contradiction, we suppose that
  $\joli S(view_{\joli Z}^I) \not\approx \joli S(view_{\joli Z}^R)$.
  That means there exists a distinguisher $\D$ such that 
  $$
    \Big| \Pr\big[\D(\joli S(view_{\joli Z}^{I}))=1\big] -
      \Pr\big[\D(\joli S(view_{\joli Z}^{R}))=1\big] \Big| 
      > \frac 1{\poly[\lambda]} \enspace.
  $$
  This allows us to define a distinguisher $\D'$ by
  $\D'(view_{\joli Z}) = \D\big(\joli S(view_{\joli Z})\big)$
  to distinguish the view of $\Z$ in the real and ideal executions, which contradicts the fact that 
  $view_{\joli Z}^I \approx view_{\joli Z}^R$. Therefore,
  we must have that
  $\joli S(view_{\joli Z}^I) \approx \joli S(view_{\joli Z}^R)$.
    Finally, since $state_{\joli X}^I\approx state_{\joli X}^R$,
  we have that 
  $$
  \joli S(view_{\joli Z}^R) \approx \joli S(view_{\joli Z}^I) =
  state_{\joli X}^I \approx state_{\joli X}^R\enspace .
  $$
   The time complexity of $\S$ is at most that of $\X$ in the real execution and of the function $f_\X$ that maps the view of $\Z$ to the public random tape of $\X$ which is assumed to be poly-time computable, and thus $\S$ is  a \ppt{} machine.
   \qed\end{proof}

Our definition of \dels{} also satisfies some form of converse of the above Theorem. We show that if a data-collector is \del{} and also protects $\Y$'s privacy in a precise manner, then it is strongly deletion-compliant. That is, we show that if we define privacy as what is implied by strong deletion-compliance -- that the view of the environment is indistinguishable in the real and ideal execution where the deletion-requester is silent -- then a private and weakly deletion-compliant data collector is strongly deletion-compliant.
\begin{definition}
  We say that a data-collector $(\X,\pi,\pi_D)$ is (computationally) $\ep$--\emph{privacy-preserving} if the view of $\Z$ in the real execution $view_\Z^R$ and the view of $\Z$ in an ideal execution $view_\Z^I$ where $\Y$ is replaced with a \emph{silent} $\Y_0$ satisfy the following. For all (\ppt{}) distinguisher $\D$,
  \begin{equation*}
    \left| \Pr[\D(view^R_\Z)=1] - \Pr[\D(view^I_\Z)=1] \right| \leq \ep\enspace .
  \end{equation*}
\end{definition}

\begin{theorem}\label{thm:priv+weak-impl-strong}
  If $(\X,\pi,\pi_D)$ is both $\ep_1$--\del{} and $\ep_2$--privacy-preserving, then it is $2(\ep_1+\ep_2)$--strongly deletion-compliant.
\end{theorem}
\begin{proof}
  By assumption of $\ep_1$--\dels{}, there exists $\S$ such that for any $\D$,
  \begin{equation*}
     \left| \Pr[\D(state^R_\X, view^R_\Z)=1] - \Pr[\D(\S(view^R_\Z),view^R_\Z)=1] \right| \leq \ep_1\enspace. 
   \end{equation*}
   Furthermore, by the assumption of $\ep_2$--privacy-preserving, $view_\Z^R\approx_{\ep_2} view_\Z^I$. It directly follows that $\S(view_\Z^R)\approx_{\ep_2}\S(view_\Z^I)$. We now argue that $\S(view_\Z^I)$ outputs a valid state for $\X$ (or something indistinguishable from a valid state). Indeed, if it wasn't the case, then one could either distinguish between $state_\X^R$ and $\S(view_\Z^R)$ or between $\S(view_\Z^R)$ and $\S(view_\Z^I)$ which is not possible except with probability at most $\ep_1+\ep_2$. The state of $\X$ produced as output of $\S(view^I_\Z)$ must be the product of an ideal execution, since $view^I_\Z$ is the view of $\Z$ in an ideal execution where $\Y$ is replaced with the silent $\Y_0$. If we let $state_\X^I$ denote the true state of $\X$ in the ideal execution, we have that
   \begin{align*}
     \Delta((state^R_\X, view^R_\Z),(state^I_\X,view^I_\Z)) 
     &\leq\Delta((state^R_\X, view^R_\Z), (\S(view^R_\Z), view^R_\Z))\\
     & \quad + \Delta((\S(view^R_\Z), view^R_\Z),(state^I_\X,view^I_\Z)) \\
     &\leq\Delta((state^R_\X, view^R_\Z),(\S(view^R_\Z), view^R_\Z)) \\
     & \quad + \Delta((\S(view^R_\Z), view^R_\Z),(\S(view^I_\Z), view^I_\Z)) \\
     & \quad+ \Delta((\S(view^I_\Z), view^I_\Z),(state^I_\X,view^I_\Z)) \\
     &\leq \ep_1+ \ep_2 + (\ep_1+\ep_2)
   \end{align*}
\qed\end{proof}

\section{Example: a Deletion-Compliant Message Board}
\label{sec:exampl-delet-compl}

We present a construction of data-collector that implements a public message board where users may post messages and fetch new messages. The data-collector presented in Fig.~\ref{fig:msg-board} is obviously not strongly deletion-compliant since the adversary can see the deletion-requester's messages. The data collector allows machines to post messages to the message board through a protocol $\pi_{\tt post}$. To post a message $m\in\bool^*$, the client machine sends it to $\X$ along with a \emph{key} $k$ used for deletion\footnote{Note that this key is not strictly necessary for deletion-compliance. $\X$ would still be deletion-compliant if anyone could delete every message.}. This key is used for uniquely identifying message and for authentication when deleting from the message board and is stored alongside the message in $\tt list$. Letting the client maching choose $k$ instead of $\X$ lets us use the results of Section~\ref{sec:hist-indep-data} to show \dels{} of $\X$.  To fetch the list of posted messages, a machine invokes $\pi_{\tt fetch}$, $\X$ then sends the list of messages contained in the list (without the associated keys) in order of insertion. To delete using $\pi_D$, $\X$ searches the list for the key $k$ and removes from the list the entry that contains this key.

\begin{figure}[h]
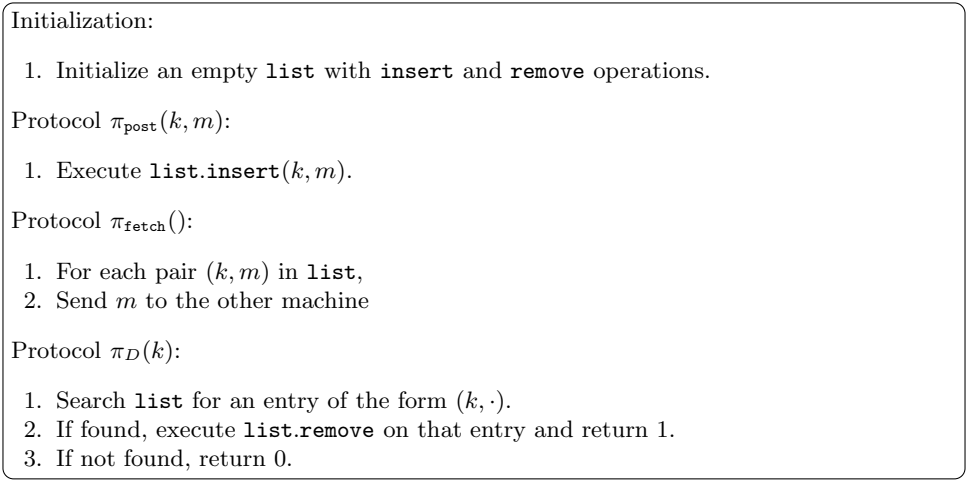
\centering

  \tikz{\node [draw,rounded corners] {
      \begin{minipage}{0.9\linewidth}
        Initialization:
        \begin{enumerate}
        \item Initialize an empty $\tt list$ with $\tt insert$ and $\tt remove$ operations.
        \end{enumerate}

        Protocol $\pi_{\tt post}(k,m)$:
        \begin{enumerate}
        \item Execute ${\tt list.insert}(k,m)$.
        \end{enumerate}

        Protocol $\pi_{\tt fetch}()$:
        \begin{enumerate}
        \item For each pair $(k,m)$ in $\tt list$, 
        \item Send $m$ to the other machine
        \end{enumerate}
        
        Protocol $\pi_D(k)$:
        \begin{enumerate}
        \item Search $\tt list$ for an entry of the form $(k,\cdot)$.
        \item If found, execute $\tt list.remove$ on that entry and return 1.
        \item If not found, return 0.
        \end{enumerate}

      \end{minipage}
    };}
\caption{The ``message board''data-collector $\X$.}
\label{fig:msg-board}
\end{figure}

A seemingly trivial, yet unfruitful approach to constructing the simulator $\S$ of Definition~\ref{def:weak del} is to argue that $\Z$ receives complete information on the state of $\X$ when protocol $\pi_{\tt fetch}$ is executed. However, the simulator must depend only on $\X$ and work for any environment, and in particular for $\Z$ that never invokes $\pi_{\tt fetch}$. 

\begin{theorem}\label{thm:example-collector}
  If $\tt list$ is a history-independent implementation of a list data structure, then the data-collector of Fig.~\ref{fig:msg-board} is \del{}.
\end{theorem}

We can use the techniques of Section~\ref{sec:hist-indep-data} to prove that this data collector satisfies \dels{} (Definition~\ref{def:weak del}), so we defer the proof to that section. Nevertheless, it is instructive to explicitly describe a simulator for the data-collector $\X$ and sketch why it satisfies Definition~\ref{def:weak del}.
We construct $\S$ as follows:
\begin{enumerate}
\item Initalize an empty $\tt list$.
\item For each protocol $\tau$ in $view_\Z$, 
  \begin{itemize}
  \item if $\tau=\pi_{\tt fetch}$, do nothing.
  \item if $\tau=\pi_{\tt post}$, let $(k,m)$ be the message sent by $\Z$ to $\X$; append $(k,m)$ to $\tt list$.
  \item if $\tau=\pi_{D}$, let $k$ be the key sent by $\Z$; if $\tt list$ has an entry of the form $(k,\cdot)$, execute ${\tt list.remove}$ to remove that entry.
  \end{itemize}
\item The output of $\S$ is $\tt list$.
\end{enumerate}
The above simulator works because if $\tt list$ is instantiated with a history-independent implementation, then there is no distinction in the state of $\X$ between the case where each invocation of $\pi_{\tt post}(k,m)$ is \emph{ultimately} followed by the corresponding $\pi_D(k)$ and the case where $\pi_{\tt post}(k,m)$ is not invoked at all.

\section{Properties of \DELs{}}
\label{sec:concrete-systems}

Our weaker definition of deletion-compliance has several appealing properties. In this section, we show the many ways in which \dels{} is composable.

\subsection{Composability of Deletion-Requesters and Data-Collectors}
\label{sec:prop-compo}

We first show that the \dels{} error scales linearly with the number of instances of $\pi$ initiated by $\Y$. As in~\cite{GGV}, we call an execution \emph{$k$--representative} if the deletion-requester $\Y$ initiates $k$ instances of $\pi$ with $\X$.
 
\begin{definition}
  A deletion-requester is \emph{$k$-representative}
  for $k\in\N$ if, when interacting with a data-collector with protocols
  $(\pi,\pi_D)$, it initiates at most $k$ instances of $\pi$.
\end{definition}

\begin{theorem}\label{thm:compo-delreq}
 Let $(\X,\pi,\pi_D)$ be a data-collector with
  $1$-representative \dels{} error $\ep_1$. Then for all
  $k\in \N$, the $k$-representative \dels{} error $\ep_k$ of $\X$
  is at most $2k\cdot \ep_1$.
\end{theorem}
\begin{proof}
  Fix $\lambda\in \N$. We prove by induction on $k$ that 
  $\ep_k(\lambda) \leq 2k\ep_1(\lambda)$. We omit $\lambda$
  for the rest of the proof.
 Let $k\geq 2$ and assume that 
  $\ep_{k-1}\leq (k-1)\ep_1$.
  Without loss of generality, for any $k$-representative
  deletion-requester $\Y$, we can assume $\Y$ is composed of $k$ ITIs
  $\Y_1,\ldots,\Y_k$ allowed to read and write on each other's
  incoming and outgoing tapes where each $\Y_i$ initiates $\pi$ with $\X$ at
  most once.

  We can represent the execution with $k$--representative $\Y$ and environment $\Z$ as an execution with $1$--representative $\Y_k$ and a new environment $\Z'$
  so that in the new execution, $\Y_k$ interacts with $\X$ and $\Y_1,\ldots,\Y_{k-1}$ are part 
  of the environment. Let $\textsc{exec}_1$ denote this $1$--representative execution.
  In particular, the view of $\Z'$ contains the view of $\Z$ and
  the views of $\Y_1,\ldots,\Y_{k-1}$.

  Consider a simulator $\S''$ such that 
  \begin{equation}\label{eq:prob1 borne}
   \max_{\Z,\Y,\D} \min_{\Ss'} \PR{\D,{state_\X,view_{\Z'}},{\Ss'(view_{\Z'}),view_{\Z'}}}
  \end{equation}
  attains the minimum. By assumption, the value of \eqref{eq:prob1 borne} is the 1--representative error of $\X$, which is exactly $\ep_1$.

  Now,  consider the $k$--representative error
  \begin{equation}\label{eq:S'}
    \ep_k=\max_{\Z,\Y,\D}\min_\Ss\Big| \Pr\big[\D(state_\X,view_\Z)=1\big] - \Pr\big[\D(\Ss(view_\Z),view_\Z)=1\big]\Big|\enspace .
  \end{equation}
 By the triangle inequality,  \eqref{eq:S'}  is upper-bounded by
   \begin{align}
     \ep_k&\leq \max_{\Z,\Y,\D}\PR{\D,{state_\X,view_\Z},{\Ss'(view_{\Z'}),view_\Z}}\label{eq:minS'} \\
        & \quad\qquad {}+ \max_{\Z,\Y,\D}\min_{\S}\PR{\D,{\Ss'(view_{\Z'}),view_\Z},{\Ss(view_\Z),view_\Z}}\nonumber\\
               &\leq \ep_1 +  \max_{\Z,\Y,\D}\min_{\Ss}\PR{\D,{\Ss'(view_{\Z'}),view_{\Z}},{\Ss(view_\Z)},view_{\Z}}\nonumber 
   \end{align}
   where the last inequality follows from~(\ref{eq:prob1 borne}) and the fact that giving the distinguisher less information in~(\ref{eq:minS'}), i.e. the view of $\Z$ instead of $\Z'$, can only decrease the distinguishing probability.

  We now bound the remaining part of the previous equation.
  Let $\Oo$ be the output of $\Ss'(view_{\Z'}^\X)$. We observe the two following facts about $\Oo$:
  \begin{enumerate}
  \item By construction of $\Ss'$, $\Oo$ is indistinguishable from the state of some data collector resulting from a valid execution.
  \item $view_{\Z'}^\X$ contains exactly the interactions between
    $\Y_1,\ldots,\Y_{k-1},\Z$ and $\X$. In particular, it has access to the views of the 1-representative deletion-requesters $\Y_1,\dots,\Y_{k-1}$, so it has the complete view of a valid $(k-1)$--representative execution.
  \end{enumerate}
  From the two above observations, we deduce that the output $\Oo$ of
  $\Ss'$ is indistinguishable from the state of the state of the data collector obtained as the
  result of a $(k-1)$--representative execution where $\Y_1,\dots,\Y_{k-1}$ and the environment $\Z$ interact with $\X$.
  Note that we do not require  this execution
  to be exactly as the $1$--representative execution $\textsc{exec}_1$ defined above, but only that it is a valid $(k-1)$--representative execution.
  Let $state_\X'$ be the state of $\X$ in this execution.
  As we have argued, $\Oo \approx_{\ep_1} state_\X'$ corresponds to the state of $\X$ in a $(k-1)$-representative execution. It follows from the inductive hypothesis that there exists a simulator $\Ss$ for any $(k-1)$--representative execution, and so by the triangle inequality
  \begin{equation}\label{eq:S'2}
    \max_{\Z,\Y,\D}\min_{\S}\PR{\D,{\Ss'(view_{\Z'}),view_\Z},{\Ss(view_\Z),view_\Z}}\leq \ep_1+ \ep_{k-1}\enspace .
  \end{equation}

  Combining the two bounds for~\eqref{eq:minS'}
  and~\eqref{eq:S'2}, we obtain
  $$
    \PR{\D,{state_\X,view_\Z},{\Ss(view_\Z),view_\Z}} \leq 2\ep_1 + \ep_{k-1} \leq 2k\ep_1,
  $$
  where we used the induction hypothesis $\ep_{k-1} \leq 2(k-1)\ep_1$.
\qed\end{proof}

For data-collectors, there are two natural ways of combining them: in parallel where $\Y$ may interact with either $\X_1$ or $\X_2$, and sequentially where $\Y$ interacts with $\X_1$ and $\X_1$ interacts with another data-collector $\X_2$ as would be the case when delegating parts of its storage or computation. 

\begin{definition}
  Given two data-collectors $(\joli X_j,\pi_j,\pi_{j,D})$, $j=1,2$, we define
  the \emph{parallel composition} $(({\cal X}_1,{\cal X}_2),(\pi_1,\pi_2),
  (\pi_{1,D},\pi_{2,D}))$ of ${\cal X}_1$ and ${\cal X}_2$ as the data-collector
  with the following behavior:
  \begin{itemize}
  \item the data-collector $(\joli X_1,\joli X_2)$ runs an instance of each
    machine $\joli X_1$ and $\joli X_2$, each running independently of one
    another (i.e. having no access to each other's incoming/outgoing tapes);
  \item when $\cal Y$ or $\cal Z$ initiates an instance of $\pi_i$ (resp.
    $\pi^D_i$), $i\in \{1,2\}$, with $({\cal X}_1,{\cal X}_2)$, it runs protocol
    $\pi_i$ (resp. $\pi^D_i$) with machine $\joli X_i$;
  \item the state of $(\joli X_1,\joli X_2)$ consists of the state of machines
    $\joli X_1$ and $\joli X_2$: $state_{(\joli X_1,\joli X_2)}= (state_{\joli
      X_1},state_{\joli X_2})$.
  \end{itemize}
\end{definition}

\begin{theorem}\label{theo:par comp}
  If $(\joli X_1,\pi_1,\pi_{1,D})$ and $(\joli X_2,\pi_2,\pi_{2,D})$
  are statistical (resp. computational) \del{} with error $\ep_1$ and
  $\ep_2$, respectively, then their parallel composition is also
  statistical (resp. computational) \del{} with error $\ep_1+\ep_2$.
\end{theorem}
\begin{proof}
  The proof is identical for statistical and computational flavors of \dels{}.
  The interactions of $\Y$ and $\Z$ with $(\X_1,\X_2)$ in the real execution can be represented by an equivalent execution with $\X_1$ and $\X_2$:
  
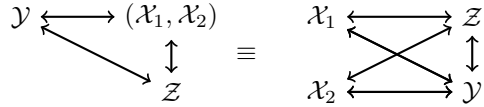
\begin{center}
  \begin{tikzpicture}[thick,<->]
    \begin{scope}[]
      \node (Y) at (0,1) {$\cal Y$};
      \node (X12) at (2,1) {$(\joli X_1,\joli X_2)$};
      \node (Z) at (2,0)  {$\joli Z$};
      \draw (Y) edge (X12) ;
      \draw (Y) edge (Z);
      \draw (X12) edge (Z);
    \end{scope}

    \node at (3,.5) {$\equiv$};
    
    \begin{scope}[xshift=4cm]
    \node (Y) at (2,0) {$\cal Y$};
    \node (X1) at (0,1) {$\joli X_1$};
    \node (X2) at (0,0) {$\joli X_2$};
    \node (Z) at (2,1)  {$\joli Z$};
    \draw (Y) edge (X1);
    \draw (Y) edge (X2);
    \draw (Y) edge (Z);
    \draw (X2) edge (Z);
    \draw (X1) edge (Z);
    \draw (X2) edge (Y);
    \draw (X1) edge (Y);
  \end{scope}
  \end{tikzpicture}
  \end{center}

  We can give two interpretation to this execution to invoke the
  \dels{} of $\X_i$: we can combine $\X_2$ to the environment and
  $\X_1$ becomes the data-collector, or we can combine $\X_1$ to the
  environment and $\X_2$ becomes the data-collector.
  
  Let $\Z'$ denote the new environment in both executions.
In the first case, we see that $view_{\Z'}^{\X_1}=view^{\X_1}_{(\X_2,\Z)} = view^{\X_1}_{\Z}$, 
  because $\X_1$ and $\X_2$ do not interact together. Similarly, in the
  second case, we have $view_{\Z'}^{\X_2}=view^{\X_2}_{(\X_1,\Z)} = view^{\X_2}_{\Z}$.
  By the \dels\ property of $\X_1$ and $\X_2$, there exist two
  simulators $\Ss_1$ and $\Ss_2$ such that 
  $$
    \Ss_1(view^{\X_1}_{(\X_2,\Z)}) \approx_{\ep_1} state_{\X_1} 
    \qquad\hbox{ and }\qquad
    \Ss_2(view^{\X_2}_{(\X_1,\Z)}) \approx_{\ep_2} state_{\X_2}\enspace .
  $$
  Since we have $view^{\X_j}_{(\X_{3-j},\Z)} = view^{\X_j}_{\Z}$ for $j\in\{1,2\}$,
  and because $view^{\X_j}_{\Z}$ is contained in $view^{(\X_1,\X_2)}_{\Z}$,
  we can construct a simulator $\S$ that computes $view^{\X_j}_{\Z}$ from $view^{(\X_1,\X_2)}_{\Z}$, and outputs
  $$
    \Ss(view^{(\X_1,\X_2)}_{\Z}) = \big(\Ss_1(view^{\X_1}_{\Z}), \Ss_1(view^{\X_2}_{\Z})\big)\enspace .
  $$
  By the triangle inequality, we have
  \begin{align*}
    &\PR{\D,{state_\X,view_\Z},{\Ss(view_\Z),view_\Z}}\\
    &= \Big| \Pr[\D((\Ss_1(view^{\X_1}_{\Z}),\Ss_1(view^{\X_2}_{\Z}),view_\Z)=1] \\
    &\qquad - \Pr[\D((state_{\X_1}, state_{\X_2}),view_\Z)=1] \Big|\\
    &\leq  \Big| \Pr[\D((\Ss_1(view^{\X_1}_{\Z}),\Ss_1(view^{\X_2}_{\Z})),view_\Z)=1]\\
    &\qquad - \Pr[\D((state_{\X_1}, \Ss_1(view^{\X_2}_{\Z})),view_\Z)=1] \Big|\\
    &\quad + \Big| \Pr[\D((state_{\X_1}, \Ss_1(view^{\X_2}_{\Z})),view_\Z)=1] \\
    &\qquad- \Pr[\D((state_{\X_1}, state_{\X_2}),view_\Z)=1] \Big|\\
    &\leq \ep_1+\ep_2
 \end{align*}
  as required.
\qed\end{proof}

Sequential composition of data-collectors is defined as follows. 

\begin{definition}
Given two data-collectors $(\joli X_j,\pi_j,\pi_{j}^{D})$, $j=1,2$, we
define the \emph{sequential composition} $({\cal X}_2\circ{\cal
X}_1,\pi_2\circ\pi_1, \pi_{2}^{D}\circ\pi_{1}^{D})$ of ${\cal X}_1$ and
${\cal X}_2$ as the data-collector with the following behavior:
\begin{enumerate}
  \item the data-collector $\joli X_2\circ\joli X_1$ runs an instance of
    each machine $\joli X_1$ and $\joli X_2$, each running independently
    of one another (i.e. having no access to each other's 
    tapes);
  \item whenever $\cal Y$ or $\cal Z$ initiates an instance of
    $\pi_2\circ\pi_1$ with machine $\joli X_2\circ \joli X_1$, run
    protocol $\pi_1$ between machine $\joli X_1$ and $\cal Y$ (or
    $\cal Z$) which in turn initiates an instance of $\pi_2$ between
    $\joli X_1$ and $\joli X_2$;
  \item the state of $\joli X_2\circ\joli X_1$ consists of the state of
    machines $\joli X_1$ and $\joli X_2$:
    $state_{\joli X_2\circ\joli X_1}= (state_{\joli X_1},state_{\joli
      X_2})$.
\end{enumerate}
\end{definition}

Our goal is to show that if $\X_1$ and $\X_2$ are both \del, then their sequential composition is \del{}. To do so, we need to reconstruct the state of both $\X_1$ and $\X_2$ using their respective simulators. The difficulty is that in the case of $\X_2\circ\X_1$, the simulator only has access to the view of $\Z$, and not the view of the interactions between $\X_1$ and $\X_2$. This view may contain elements that are necessary to reconstruct the individual states, but that are inaccessible to the simulator. To circumvent this issue, we add the restriction that the view between $\X_1$ and $\X_2$ for the protocol sessions that were first initiated by $\Z$ can be simulated by $\S$ having only access to the view of $\Z$. We call property $\ep$--\emph{independence} if this view can be reconstructed with error $\ep$. It is satisfied for example when the messages between $\X_1$ and $\X_2$ in a session of $\pi_2$ that follows $\pi_1$ initiated by $M$ is an efficient function of only the previous interactions of $M$ and $\X_1$. In this scenario, seeing all of the interactions between $M$ and $\X_1$ is sufficient to replicate the resulting interactions between $\X_1$ and $\X_2$.
\begin{definition}\label{def:independent-comp}
  Let $M$ be an arbitrary ITI and let $view_{\X_1\mid M}^{\X_2}$ denote the view between $\X_1$ and $\X_2$ restricted to the instances of $\pi_2$ that follow an instance of $\pi_1$ initiated by $M$. 
  We say that the sequential composition $\X_2\circ\X_1$ of $\X_1$ and $\X_2$ is (computationally) $\ep$--\emph{independent} if there exists a \ppt{} machine $\mathcal{V}$ such that for all ITI $M$ interacting with $\X_1$ using protocol $\pi_1$, 
  and for all (\ppt{}) distinguisher $\D$,
  \begin{equation*}
    \left| \Pr[\D(\mathcal{V}(view_M^{\X_1}))=1]-\Pr[\D(view_{\X_1\mid M}^{\X_2})=1] \right|\leq \ep(\lambda)\enspace .
  \end{equation*}
\end{definition}

\begin{theorem}\label{theo:seq comp}
  If $(\X_1,\pi_1,\pi_{1}^{D})$ and $(\X_2,\pi_2,\pi_{2}^{D})$ are
  (computationally) \del\ data-collectors with respective error $\ep_1$ and $\ep_2$, then the (computationally) $\ep$--independent sequential composition 
  $(\X_2\circ \X_1, \pi_2\circ \pi_1, \pi_{2}^{D}\circ \pi_{1}^{D})$
  is (computationally) \del{} with error $\ep_1+\ep_2+2\ep$.
\end{theorem}
\begin{proof}
  We construct a simulator $\Ss$ such that given the view
  $view_\Z^{\X_2\circ \X_1}$, produces an output indistinguishable
  from $(state_{\X_1}, state_{\X_2})$. To do this, we rely on the \dels\ property of $\X_j$ for $j=1,2$ 
  and assume the existence of two simulators $\Ss_1$ and $\Ss_2$ such that $\Ss_j$ can recreate the state of $\X_j$ in an execution from the view of the environment. We now show how  to simulate the state of $\X_2 \circ \X_1$ from the view of $\Z$ using $\Ss_1$ and $\Ss_2$.

  \begin{figure}[h]
    \centering
    \begin{subfigure}[b]{0.3\textwidth}
      \centering
      \begin{tikzpicture}[thick,<->]
        \begin{scope}[]
          \node (Y) at (0,2) {$\cal Y$};
          \node (X1) at (2,2) {$\X_1$};
          \node (X2) at (3.5,2) {$\X_2$};
          \node (Z) at (2,0) {$\Z$.};
          \draw (Y) edge (X1);
          \draw (X1) edge (X2); 
          \draw (Y) edge (Z);
          \draw (Z) edge (X1);
          \draw (Z) edge (X2);
          \draw (1.75,1.75) rectangle (3.75,2.25);
        \end{scope}
      \end{tikzpicture}
      \caption{Real execution with data collector $\X_2\circ \X_1$.}
      \label{fig:simX1A}
    \end{subfigure}
    \hfill
    \begin{subfigure}[b]{0.3\textwidth}
      \centering
      \begin{tikzpicture}[thick,<->]
        \begin{scope}[]
          \node (Y) at (0,2) {$\cal Y$};
          \node (X1) at (2,3) {$\X_1^{\Y}$};
          \node (X11) at (2,1.5) {$\X_1^{\Z}$};
          \node (X2) at (3.5,2) {$\X_2$};
          \node (Z) at (2,0) {$\Z$,};
          \draw (Y) edge (X1);
          \draw (X1) edge (X2); 
          \draw (Y) edge (Z);
          \draw (Z) edge (X11);
          \draw (Z) edge (X2);
          \draw (X1) -- (X11);
          \draw (X11) -- (X2);
          \draw (1.75,1.25) rectangle (2.25,3.25);
        \end{scope}
      \end{tikzpicture}
      \caption{Simulate $\X_1$ with two new ITIs.}
      \label{fig:simX1B}
    \end{subfigure}
    \hfill
    \begin{subfigure}[b]{0.3\textwidth}
      \centering
      \begin{tikzpicture}[thick,<->]
        \begin{scope}[]
          \node (Y) at (0,3) {$\cal Y$};
          \node (X1) at (2,3) {$\X_1^{\Y}$};
          \node (X11) at (2,1.5) {$\X_1^{\Z}$};
          \node (X2) at (3.5,2) {$\X_2$};
          \node (Z) at (2,0) {$\Z$,};
          \draw (Y) edge (X1);
          \draw (X1) edge (X2);
          \draw (Y) edge (Z);
          \draw (Z) edge (X11);
          \draw (Z) edge (X2);
          \draw (X1) -- (X11);
          \draw (X11) -- (X2);
          \draw (-.25,2.75) rectangle (2.25,3.25);
        \end{scope}
      \end{tikzpicture}
      \caption{Defining a new deletion requester and environment.}
      \label{fig:simX1C}
    \end{subfigure}
    \caption{Simulating the state of $\X_1$ for sequential composition.}
    \label{fig:simX1}
  \end{figure}

  We begin by constructing the state of $\X_2$ using $\Ss_2$. Observe
  that in the real execution (Fig.~\ref{fig:simX1A}), we can divide
  $\X_1$ into two parts $\X_1^{\Y}$ and $\X_1^{\Z}$ as follows
  (Fig.~\ref{fig:simX1B}). Machine $\X_1^{\Y}$ interacts only with
  $\Y$ and machine $\X_1^{\Z}$ only with $\Z$. We allow $\X_1^\Y$ and
  $\X_1^\Z$ to arbitrarily message each other to faithfully replicate
  the behavior of $\X_1$ from $\Y$'s and $\Z$'s points of view. Note
  that neither machine knows whether it is interacting with $\Y$ or
  $\Z$ since in reality, it is interacting with ITIs under their
  control, but this is not necessary for the proof; we only require
  that those machines are well-defined once we fix $\Y$ and $\Z$.  We
  can combine $\X_1^{\Y}$ with $\Y$ to define a new deletion-requester
  $\Y'$, and $\X_1^{\Z}$ with $\Z$ to define a new environment
  $\Z'$. Note that this corresponds to a valid execution with data
  collector $\X_2$, deletion requester $\Y'$ and environment $\Z'$
  because any protocol $\pi_2$ initiated by $\Y'$ will be followed by
  the corresponding $\pi_2^D$ in the terminal phase since those are
  exactly the protocols $\pi_2$ that follow the $\pi_1$'s initiated by
  $\Y$. By assumption there exists a simulator $\Ss_2$ that can
  simulate the state of $\X_2$ in any execution from the view of the
  environment $\Z'$. We thus have that
  $\Ss_2(view_{\Z'}^{\X_2})\approx state_{\X_2}$. By the assumption
  that the composition is $\ep$--independent, $view_{\Z'}^{\X_2}$,
  which consists of the view of any protocol between $\X_1^{\Z}$ and
  $\X_2$ initiated by $\Z$, corresponds to $view_{\X_1\mid \Z}^{\X_2}$ from Definition~\ref{def:independent-comp}, and can thus be (approximately) computed
  efficiently from $view_\Z^{\X_1}$ using a function $\mathcal{V}_1$.

  To simulate the state of $\X_1$ using $\Ss_1$, the idea is
  similar. Going back to the real execution (Fig.~\ref{fig:simX1A}),
  we can separate $\X_2$ in two parts: $\X_2^{\Y}$ receives the
  requests from $\X_1$ originating from $\Y$ and $\X_2^{\Z}$ receives
  those from $\X_1$ originating from $\Z$ (again, neither machine
  knows which party initiated the session of $\pi_1$ that resulted in
  its interaction with $\X_1$).  The two components can communicate in
  the interest of simulating faithfully the execution of $\X_2$.
  This alternative execution is depicted in Fig.~\ref{fig:simX2A}.
  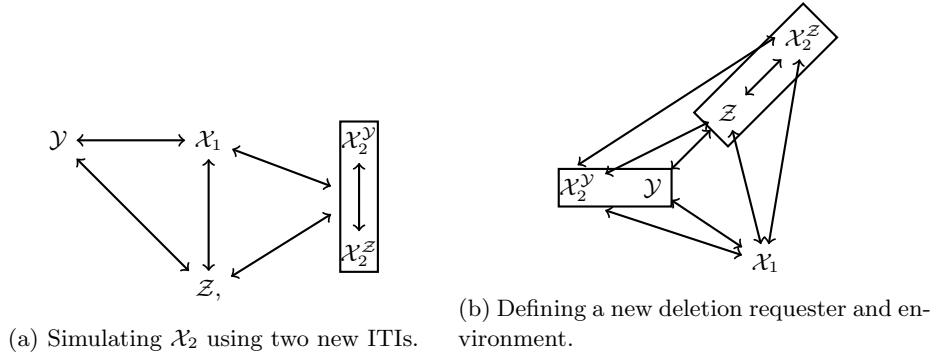
\begin{figure}[h]
    \centering
    \begin{subfigure}[b]{0.45\linewidth}
      \centering
      \begin{tikzpicture}[thick,<->]
        \begin{scope}[]
          \node (Y) at (0,2) {$\cal Y$};
	  \node (X1) at (2,2) {$\X_1$};
          \node (X22) at (4,.5) {$\X_2^{\Z}$};
          \node (X21) at (4,2) {$\X_2^{\Y}$};
	  \node (X3) at (4,1.25) {$\phantom{\X_2^{\Z}}$};
          \node (Z) at (2,0) {$\Z$,};
          \draw (Y) edge (X1);
          \draw (Y) edge (Z);
          \draw (Z) edge (X1);
          \draw (Z) edge (X3);
	  \draw (X1) -- (X3);
	  \draw (X21) -- (X22);
          \draw (3.75,.25) rectangle (4.25,2.25);
        \end{scope}
      \end{tikzpicture}
      \caption{Simulating $\X_2$ using two new ITIs.}
      \label{fig:simX2A}
    \end{subfigure}
    \begin{subfigure}[b]{0.45\linewidth}
      \centering
      \begin{tikzpicture}[thick,<->]
        \begin{scope}[]
	  \node (X21) at (0,2) {$\X_2^{\Y}$};
	  \node (Y) at (1,2) {$\Y$};
	  \node (X1) at (2.5,1) {$\X_1$};
	  \node (Z) at (2,3) {$\Z$};
	  \node (X22) at (3,4) {$\X_2^{\Z}$};
	  \draw (Z) -- (X22);
	  \draw (X1) -- (Z);
	  \draw (X1) -- (Y);
	  \draw (X1) -- (X21.south east);
	  \draw (X1) -- (X22);
	  \draw (Y) -- (Z);
	  \draw (Z) -- (X21);
	  \draw (X21.north) -- (X22.west);
	  \draw (-.25,1.75) rectangle (1.25,2.25);
	  \draw (2,2.55)--(1.55,3)--(3,4.45)--(3.45,4)--cycle;
        \end{scope}
      \end{tikzpicture}
      \caption{Defining a new deletion requester and environment.}
      \label{fig:simX2B}
    \end{subfigure}
    \caption{Simulating the state of $\X_2$ for sequential composition.}
  \end{figure}
  We now define a new environment $\Z'=(\Z,\X_2^{\Z})$ and a new
  deletion-requester $\Y'=(\Y,\X_2^\Y)$
  (Fig.~\ref{fig:simX2B}). The new machines $\Y'$ and $\Z'$ now interact with $\X_1$ 
  through the protocols $\pi_2\circ \pi_1$ and
  $\pi_{2}^{D}\circ \pi_{1}^{D}$, e.g. when $\Y'$ initiates $\pi_1$ with $\X_1$, $\X_1$ initiates a session of $\pi_2$ with $\Y'$, and after $\pi_2$ concludes, they finish the session of $\pi_1$.
Since every protocol between
  $\Y'$ and $\X_1$ is initiated by $\Y$ (and thus followed by the
  corresponding deletion request), machine $\Y'$ still satisfies the
  definition of a deletion requester. 
  Since $\X_1$ is \del, there exists a simulator $\Ss_1$
  such that $\Ss_1(view_{\Z'}^{\X_1}) \approx state_{\X_1}$.
  Note that $view_{\Z'}^{\X_1}$ can be efficiently computed from
  $view_\Z^{\X_1}$ by the assumption of $\ep$--independence: $view_{\Z'}^{\X_1}$
  corresponds to $view_{\X_2\mid \Z}^{\X_1}$, so there is an efficient algorithm
  $\mathcal{V}_2$ such that $\mathcal{V}_2(view_\Z^{\X_1})\approx_\ep
  view_{\Z'}^{\X_1}$.
  
  To conclude the proof, we define the simulator $\Ss$ that,
  given $view_\Z^{\X_2\circ \X_1}$, it first computes 
  $view_{(\Z,\X_2^{\Z})}^{\X_1}$ and $view_{(\X_1^{\Z},\Z)}^{\X_2}$ using $\mathcal{V}_1$ and $\mathcal{V}_2$, then applies the corresponding simulator: on input $view_\Z$, $\S$ outputs  
  \begin{equation*}
  \big(\Ss_1(\mathcal{V}_1(view_{\Z})), \Ss_2(\mathcal{V}_2(view_{\Z}))\big)\approx_{\ep'}state_{\X_2\circ\X_1}
\end{equation*}
where by the triangle inequality, $\ep'\leq \ep_1+\ep_2+2\ep$ 
\qed\end{proof}

\subsection{History Independence Implies \DELs{}}
\label{sec:hist-indep-data}

History independence is a concept introduced by Naor and Teague~\cite{NT-history} to capture the notion that the memory representation of an abstract data structure (ADS) does not reveal any more information than what can be inferred from the content or state of the data structure. In particular, a history independent implementation does not reveal the history of operations that lead to its current state. 

As a simple example of a history independent data structure, consider the abstract data structure consisting of a set with insertion and deletion.  A history independent implementation of this ADS could consist of keeping the elements from the set in a sorted list. The memory representation (the list) is uniquely determined by the state of the ADS (the elements in the set). More generally, any implementation of an ADS that has a \emph{canonical} representation for each ADS state is history independent~\cite{hartline_characterizing_2005}.

The original work on deletion-compliance~\cite{GGV} uses the example of a history independent dictionary from~\cite{NT-history} to build a strongly deletion-compliant data-collector that acts as a store of information. 
We show that when we drop the privacy requirements of strong deletion-compliance and instead adopt the weaker definition, history independence is a sufficient condition to obtain
\dels{} for a broad class of data-collectors. We will use the more general 
definition of history independence for arbitrary ADS
of~\cite{hartline_characterizing_2005}.

An abstract data structure can be represented as a graph where the nodes are the possible states of the ADS and the edges are operations that send the ADS from one state to the next. Let $A$ and $B$ be states of an ADS. A sequence of operations $S$ (i.e. a path in the graph) that takes state $A$ to state $B$ is denoted by $A\tounder S B$. The notation ``$a\in A$'' means that $a$ is a \emph{memory representation} of ADS state $A$ for some implied implementation of the ADS. We let $\Pr[a\tounder S b]$ denote the probability that, starting from memory representation $a$ of state $A$, the sequence of operations $S$ maps to the memory representation $b$ of state $B$

\begin{definition}[Strong History Independence~\cite{hartline_characterizing_2005}]\label{def:history-independence}
  An ADS implementation is strongly history
  independent if, for any two sequences of operation
  $S$ and $T$ that take the data structure from state $A$
  to state $B$, the distribution over memory representations after $S$
  is identical to the distribution after $T$. That is,
  $$
    (A \tounder S B)
    \text{ and } 
    (A \tounder T B)
    \Longrightarrow
    \forall a\in A, \forall b\in B,\; \Pr[a \tounder S b]
    = \Pr[a \tounder T b].
  $$
\end{definition}

We define deletion in the
context of abstract data structure as follows.
\begin{definition}\label{def:del op}
  A \emph{deletion operation} of an ADS
  operation $p$ is an operation $p_D$ such that
  for all sequences of operations $R,S,T$ and
  every states $A,B$, we have that
  $A \tounder{Rp Sp_D T} B$ and
  $A \tounder{RST} B$.
\end{definition}

We now consider a data-collector $(\X,\pi,\pi_D)$ whose state transition function
can be described by an abstract data structure. We assume that
each session $s$ of protocol $\pi$ in an execution correspond to a sequence of operations $\Pi^s$ in the abstract
data structure.
Naturally, we ask that $\pi_D$ defines a sequence of operations
$\Pi_D^s$ that reverses the corresponding sequence $\Pi^s$ in the ADS, as in
Definition~\ref{def:del op}. We show that if $\X$ implements this ADS
in a history independent way, then it is \del.

The simulation strategy presented in the proof of Theorem~\ref{thm:hist-indep} requires the data collector to be deterministic. This is not so much a restriction since~\cite{hartline_characterizing_2005} have shown that strongly history independent data structures must be deterministically determined by their content, up to some random initialization (e.g. the choice of a hash function). 

\begin{theorem}\label{thm:hist-indep}
  Let $(\X,\pi,\pi_D)$ be a data-collector such that
  \begin{enumerate}
  \item $\X$ is an implementation of an abstract data structure whose states are the possible values for the work tape of $\X$,
  \item $\X$ is deterministic up to some randomness used for initialization that does change the views of the other machines,
  \item to each session of protocol $\pi$ corresponds a sequence of state transitions $T$ for the abstract data structure, and
  \item for each session of $\pi$ with state transition $T$, the corresponding session of $\pi_D$ is a deletion operation $T_D$ as defined in Definition~\ref{def:del op}.
  \end{enumerate}
  If $\X$ is history-independent, then $\X$ is \del.
\end{theorem}
\begin{proof}
  During an execution between $\X$, $\Y$ and $\Z$, a list of protocol
  sessions are initiated with $\X$. Let $\Pi = (\tau_1,\ldots,\tau_N)$
  be this list where each $\tau_i$ is the session ID of an instance of $\pi$ or $\pi_D$.
   To each $\tau_i$ we associate the
  sequence of state transitions $T_i$ that represent the change in the
  state of $\X$ resulting from the execution of $\tau_i$, i.e. such that
  \begin{equation}
    \label{eq:2}
    state_\X^{i}\tounder {T_i} state_\X^{i+1}
  \end{equation}
  where $state^k_\X$ is the state of $\X$ after the $k$th protocol $\tau_k$.
  Let $L\subset [N]$ be the indices such that $i\in L$ if and only if
  $\tau_i$ was initiated by $\Y$. Since $\Y$
  is the deletion-requester, for every $\tau_i$ ($i\in L$) that corresponds to a
  session of $\pi$, there is a $j>i$ such that $\tau_j$ is the
  corresponding session of $\pi_D$. Let $R\subset [N]\times [N]$ be
  the set of pairs $(i,j)$ with $i\in L$ such that $\tau_j$ is the session of
  $\pi_D$ that corresponds to session $\tau_i$ of $\pi$ as described
  above. For every $(i,j)\in R$, the state transitions $T_i$ and $T_j$
  corresponding to $\tau_i$ and $\tau_j$ satisfy Def.~\ref{def:del
    op}, i.e.\ $T_j$ undoes $T_i$ in the ADS.

The simulator $\S$ of Definition~\ref{def:weak del} first initializes $\X$
  with a fresh random tape and
 simulates $\X$ by invoking protocols $\pi$ and $\pi_D$ initiated by $\Z$ using the contents of the incoming and outgoing tapes from the view of $\Z$. Since $\X$ is deterministic after initialization, this simulation will lead to a state for $\X$ that is consistent with the view of $\Z$.
 Using the notation introduced
  above, $\Ss$ simulates the execution of protocol sessions $\tau_i$
  for $i\in{[N]\setminus L}$. This sequence contains the same
  protocols as in the real execution, except for the protocols $\pi$ and $\pi_D$
   initiated by $\Y$.

  We compare the simulation of $\Ss$ using $view_\Z^\X$ to the real
  execution between $\X$, $\Y$ and $\Z$. 
  In the real execution, protocol sessions $\tau_1,\dots,\tau_N$ are
  executed sequentially, resulting in state transitions
  $T_1,\dots,T_N$. Let $X$ denote the state of the ADS implemented by
  $\X$ such that $\emptyset \tounder {T_1\dots T_N} X$. By the strong
  history independence of $\X$ (Def.~\ref{def:history-independence})
  and the definition of deletion for abstract data structures
  (Def.~\ref{def:del op}), we have that for every $(i,j)\in R$ the
  sequence of state transitions
  $T_1,\dots,T_{i-1},T_{i+1},\dots,T_{j-1},T_{j+1},\dots,T_N$ also
  maps the initial state $\emptyset$ to state $X$. Therefore, if we
  let $T_{[N]\setminus L}$ denote the
  state transitions for every $i\notin L$ (i.e.\ only for the sessions
  $\tau_i$ initiated by $\Z$), we have that
  $\emptyset \tounder {T_{[N]\setminus L}} X$. By the history
  independence of $\X$, the internal representation $state_\X$ of $\X$ for the ADS
  state $X$ resulting from $T_1,\dots,T_N$ in the real execution is
  identically distributed to the state $\S(view_\Z)$ in the simulated
  execution resulting from $T_{[N]\setminus L}$. We have thus shown
  that for every distinguisher $\D$,
  $$
    \PR{\D,{state_{\X},view_\Z},{\Ss(view_{\Z}^\X),view_\Z}} =0 \enspace.
    $$
\qed\end{proof}

\subsubsection{Proof of Theorem~\ref{thm:example-collector}}
\label{sec:proof-theorem}

Using Theorem~\ref{thm:hist-indep}, it is very easy to prove that the data-collector of Section~\ref{sec:exampl-delet-compl} is \del{}. We only need to show that it satisfies all the requirements of Theorem~\ref{thm:hist-indep}. Let $\X$ be the data-collector of Fig.~\ref{fig:msg-board},
\begin{enumerate}
\item it implements a history-independent list: its state is the state of the list;
\item it is deterministic;
\item each protocol $\pi_{\tt post}$, $\pi_{\tt fetch}$ and $\pi_D$ corresponds to a (possibly empty) transition in the list abstract data structure; and
\item for some $k$, each execution of $\pi_D(k)$ triggers the ADS operation $\tt list.remove$ of the corresponding ADS operation $\tt list.insert$ triggered by $\pi_{\tt post}(k,\cdot)$.
\end{enumerate}

\section{Conclusion \& Open Questions}
\label{sec:conclusion}

We have shown that the concept of compliance to the ``right to be forgotten'' is compatible with formalisms that do not necessarily provide privacy from third parties. Under our new definition, a data collector is able to prove that it has forgotten the deletion requester's data by showing that its state is consistent with having only interacted with the other users of the system.

An interesting question is to further study the interplay of privacy and deletion-compliance. For example, if a protocol $\pi$ is private even against a malicious data collector for some appropriate notion of privacy, then does that imply that any data collector is deletion-compliant when deletion requesters interact through $\pi$? 

Despite our more permissive definition to allow a larger class of data collectors, there are still natural situations that fall outside of that class. For example, a data collector that encrypts the data it collects and merely throws away the key upon deletion should intuitively satisfy the notion of computational deletion-compliance. However, the difference in size in the state of the data collector in the cases where $\Y$'s data is present or not allows to distinguish both cases. We have proposed a potential solution to this problem -- giving the simulator a \emph{masked} view of $\Y$ -- that could form the basis of future work.

\printbibliography{}

\end{document}